%% file: IE_arXiv2.tex
\documentclass[twocolumn]{article}
\usepackage[dvipdfmx]{graphicx,xcolor}
\usepackage[fleqn]{amsmath}
\usepackage{amssymb}
\usepackage{amsthm}
\usepackage{overpic}
\usepackage{newtxtext}
\usepackage[varg]{newtxmath}
\usepackage[dvipdfm]{geometry}
\allowdisplaybreaks

\title{
  An Information-Theoretical Analysis of the Minimum Cost to Erase Information\thanks{Portions of this paper were presented at the 39th Symposium on Information Theory and Its Applications \cite{matsuta2016omc}, and at the 2017 IEEE Information Theory Workshop \cite{matsuta2017minimum}.}
}
\author{
  Tetsunao Matsuta\thanks{tetsu@ict.e.titech.ac.jp} \and Tomohiko Uyematsu\thanks{uematsu@ict.e.titech.ac.jp}  
}
\date{\empty}

\input{myedefinition}

\begin{document}
\maketitle

\renewcommand{\thefootnote}{\fnsymbol{footnote}}
\footnote[0]{The authors are with Dept.\ of Information and Communications Engineering, Tokyo Institute of Technology, Tokyo, 152-8552 Japan.}
\renewcommand{\thefootnote}{\arabic{footnote}}

{\small
  \textbf{SUMMARY}  
  We normally hold a lot of confidential information in hard disk drives and solid-state drives. When we want to erase such information to prevent the leakage, we have to overwrite the sequence of information with a sequence of symbols independent of the information. The overwriting is needed only at places where overwritten symbols are different from original symbols. Then, the cost of overwrites such as the number of overwritten symbols to erase information is important. In this paper, we clarify the minimum cost such as the minimum number of overwrites to erase information under weak and strong independence criteria. The former (resp.\ the latter) criterion represents that the mutual information between the original sequence and the overwritten sequence normalized (resp.\ not normalized) by the length of the sequences is less than a given desired value.

  \textbf{Key words:}
  data erasure, distortion-rate function, information erasure,  information spectrum, random number generation
}

\section{Introduction}
Since services and activities using various types of information have increased, we normally hold a lot of confidential information. For example, storage devices such as hard disk drives (HDDs), solid-state drives (SSDs) and USB flash drives of individuals and companies hold personal addresses, names, phone numbers, e-mail addresses, credit card numbers, etc. When we want to discard, refurbish or just increase the security of these devices, we will usually erase information to prevent the leakage.

In order to erase information, we have to overwrite the sequence of information with a sequence of symbols independent of the information. Commonly used methods of erasure are to overwrite information with uniform random numbers or repeated specific patterns such as all zeros and all ones. There are several standards \cite{usdept1995national,gutmann1996secure,schneier1996applied,usair1998air,usnatio2006special} to erase information. Although most of these standards propose to repeat overwriting many times, overwriting data once is adequate to erase information for modern storage devices (see, e.g., \cite[Section 2.3]{usnatio2006special}).

The overwriting is needed only at places where overwritten symbols are different from original symbols, e.g., 0 to 1 or 1 to 0 for binary sequences. If there are so many overwritten symbols, the overwriting damages devices, shortens the storage life and may also take write time. This is crucial for devices with a limited number of writes such as SSDs and USB flash drives. Thus, we want to reduce the number of overwritten symbols when we erase information. Here comes a natural question: ``What is the minimum number of overwritten symbols?''.

In this paper, we clarify the minimum cost such as the minimum number or time of overwrites to erase information. As we stated in the above, for a binary sequence, the overwriting occurs at places where overwritten symbols are different from original symbols. In this case, a proper measure of the cost is the Hamming distance between the original sequence and the overwritten sequence. From this point of view, the information erasure can be modeled by correlated sources as Fig.~\ref{fig:iemodel} which actually is a somewhat general model. In this model, sequences emitted from source 1 and source 2 represent confidential information and information to be erased, respectively. For example, source 1 and source 2 are regarded as a fingerprint and its quantized image, respectively. When two correlated sources are identical, the model corresponds to the above mentioned situation. As shown in this figure, the encoder can observe one of the sequences. The encoder outputs a sequence that represents the overwritten sequence. Here, we allow the encoder to observe a uniform random number of limited size to generate an independent sequence. Then, the cost can be measured by a function of the input source sequence and the output sequence of the encoder.

\input{iemodel.tex}

For this information erasure model, we consider a \textit{weak} and a \textit{strong} independence criteria. The former (resp.\ the latter) criterion represents that the mutual information between the source sequence and the output sequence of the encoder normalized (resp.\ not normalized) by the length (blocklength) of sequences is less than a given desired value. For the weak independence criterion, we consider the \textit{average} cost and the \textit{worst-case} cost. The former cost represents the expectation of the cost with respect to the sequences. The latter cost represents the limit superior in probability \cite{hanspringerinformation} of the cost. Then, by using information-spectrum quantities \cite{hanspringerinformation}, we characterize the minimum average and the minimum worst-case costs for general sources, where the block length is unlimited.  For the strong independence criterion, by employing a stochastic encoder, we give a single-letter characterization of the minimum average cost for stationary memoryless sources, where the blocklength is unlimited. On the other hand, for the strong (same as the weak in this case) independence criterion, we also consider the non-asymptotic minimum average cost for a given finite blocklength. Then, we give a single-letter characterization of it for stationary memoryless sources. We show that the minimum average and the minimum worst-case costs can be characterized by the \textit{distortion-rate function} for the lossy source coding problem (see.\ e.g., \cite{hanspringerinformation}) when the two correlated sources are identical. This means that our problem setting gives a new point of view of the lossy source coding problem. We also show that for stationary memoryless sources, there exists a sufficient condition such that the optimal method of erasure from the point of view of the cost is to overwrite the source sequence with repeated identical symbols.

There are some related studies \cite{1056749,8007053} investigating a relationship between a cost and statistical independence of sequences. These studies deal with correlated two sequences (referred to as confidential sequence and public sequence in this paper) and consider systems that reveal a sequence (referred to as revealed sequence) related to the public sequence while keeping the confidential sequence secret. In \cite{1056749}, the public sequence is encoded to a codeword and is decoded to the revealed sequence. In \cite{8007053}, the public sequence is directly and randomly mapped to the revealed sequence. These studies adopt the mutual information\footnote{More precisely, the study \cite{1056749} adopts the conditional entropy of the confidential sequence given the codeword.} between the confidential sequence and the revealed sequence (or codeword in \cite{1056749}) in order to measure the independence. 
Then, these studies give a trade-off between the mutual information normalized by the blocklength and the average distortion (i.e., cost) between the public sequence and the revealed sequence. We note that in these studies, the uniform random number of limited size is not assumed. Especially, in \cite{1056749}, the system reveals the sequence via a codeword without any auxiliary random number. Thus, system models in \cite{1056749} and \cite{8007053} are fundamentally different from our information erasure model. Moreover, these studies only consider sequences emitted from stationary memoryless sources and a certain limited distortion (cost) function. Thus, problem formulations in these studies are different especially from that for the weak independence criterion in our study. The problem formulation in the study \cite{8007053} is rather related to that for the strong independence criterion in which we consider a stochastic encoder and stationary memoryless sources.
%
However, in \cite{8007053} (and also \cite{1056749}), there is not any discussion about the optimality of the revealed sequence of repeated identical symbols which is important in the information erasure for comparison with a known method.

The rest of this paper is organized as follows. In Section \ref{sec: weak criterion}, we give some notations and formal definitions of the minimum average and the minimum worst-case costs under the weak independence criterion. Then, we characterize these costs for general sources. In Section \ref{sec: strong criterion}, we give the formal definition of the minimum average cost under the strong independence criterion. We also give the formal definition of the non-asymptotic minimum average cost. Then, we give a single-letter characterization of these costs and some results obtained from this characterization. In Section \ref{sec: proofs of theorems}, we show proofs for characterizations of minimum costs under the weak independence criterion. In Section \ref{sec: conclusion}, we conclude the paper.

\section{Minimum Costs to Erase Information under the Weak Independence Criterion}
\label{sec: weak criterion}
In this section, we consider the minimum average and the minimum worst-case costs under the weak independence criterion, and characterize these costs for general sources. We show some special cases of these costs in this section.

\subsection{Problem Formulation}
In this section, we provide the formal setting of the information erasure and define the minimum average and the minimum worst-case costs under the weak independence criterion.

Unless otherwise stated, we use the following notations throughout this paper (not just this section). The probability distribution of a random variable (RV) $X$ is denoted by the subscript notation $P_{X}$, and the conditional probability distribution for $X$ given an RV $Y$ is denoted by $P_{X|Y}$. The $n$-fold Cartesian product of a set $\cX$ is denoted by $\cX^n$ while an $n$-length sequence of symbols $(a_1,a_2,\cdots,a_n)$ is denoted by $a^n$. The sequence of RVs $\{X^n\}_{n=1}^\infty$ is denoted by the bold-face letter $\bfX$. Hereafter, $\log$ means the natural logarithm.

Let $\cX$, $\cY$ and $\hat \cX$ be finite sets, $M_n$ be a positive integer, and $\cU_{M_n}=\{1,2,\cdots,M_n\}$. Let $U_{M_n}$ be an RV uniformly distributed on $\cU_{M_n}$, and $(X^n,Y^n)$ be a pair of RVs on $\cX^n\times\cY^n$ such that $(X^n,Y^n)$ is independent of $U_{M_n}$. The pair $(\bfX,\bfY) = \{(X^n, Y^n)\}_{n = 1}^{\infty}$ of a sequence of RVs represents a pair of general sources \cite{hanspringerinformation} that is not required to satisfy the consistency condition.

For the information erasure model (Fig.\ \ref{fig:iemodel}), let $f_{n}:\cX^n\times \cU_{M_n} \ra \hat \cX^n$ be an encoder, and $c_{n}:\cX^n\times\hat \cX^n \ra [0,\infty)$ be a cost function satisfying
\begin{align*}
    \sup_{n \geq 1}\sup_{(x^n,\hat x^n)\in\cX^n\times\hat\cX^n}c_{n}(x^n,\hat x^n)
    \teq c_{\rm max} < \infty.
\end{align*}

We give two examples of the information erasure model to better understand it.
\begin{ex}
    \label{ex: example 1}
    Let a sequence $Y^n$ be confidential $n$-length binary data and be observed by some reading device, where we define $\cY \teq \{0, 1\}$. Let a sequence $X^n$ be the observed $n$-length binary data which is actually stored in a storage device, where we define $\cX \teq \{0, 1\}$. Now suppose that we can no longer read $Y^n$, but we can access the storage device and read the stored data $X^n$. Then, we want to overwrite $X^n$ to keep $Y^n$ secret. To this end, let us overwrite the data by all zero sequence. Then, we can define $\hat \cX \teq \{0,1\}$ and the encoder as $f_{n}(x^n, u) \teq (0,0,\cdots, 0)$ for any $x^n \in \cX^n$ and any $u \in \cU_{M_n}$. If we only overwrite a half of the data, i.e., we define the encoder as $f_{n}(x^n, u) \teq (x_1,x_2,\cdots, x_{n/2}, 0, 0, \cdots, 0)$ for any $x^n \in \cX^n$ and any $u \in \cU_{M_n}$, the output of the encoder is no longer independent of $Y^n$, but a cost may be reduced. Obviously, we can define a more complicated encoder as follows: Let $M_n = 2$ and
    \begin{align*}
        f_{n}(x^n, u) \teq
        \begin{cases}
            (0,0,\cdots, 0) & \mbox{ if } x_1 = 0, u = 1,\\
            (1,1,\cdots, 1) & \mbox{ if } x_1 = 1, u = 1,\\
            (1,1,\cdots, 1) & \mbox{ if } x_1 = 0, u = 2,\\
            (0,0,\cdots, 0) & \mbox{ if } x_1 = 1, u = 2.
        \end{cases}
    \end{align*}
    If we wish to count the number of overwrites of binary data, we define the cost function by the (normalized) hamming distance, i.e., $c_{n}(x^n, \hat x^n) \teq \frac{1}{n} \sum_{i = 1}^{n}\textbf{1}\{x_{i} \neq \hat x_{i}\}$, where $\mathbf{1}\{\cdot\}$ denotes the indicator function.
\end{ex}
\begin{ex}
    Let $Y^n$ be a confidential grayscale image with rather large $n$ dots, and $X^n$ be its quantized binary image\footnote{$0$ and $1$ represent black and white dots, respectively.} printed on a \textit{paper}, where we define $\cY \teq \{0, 1, 2, \cdots 255\}$ and $\cX \teq \{0, 1\}$. When we discard the paper of the binary image $X^n$, we modify\footnote{When shredding the paper into strips, it may be reassembled. Thus, we want to modify the original image.} it by using an eraser and a black ink pen in order to keep the grayscale image $Y^n$ secret. If the eraser can erase black dots clearly (probably the eraser or the black ink is special), the modified image is also a binary image. Thus, we can define $\hat \cX = \{0, 1\}$ and encoders as those in Example \ref{ex: example 1}. Suppose that the eraser is more expensive than the pen, and we pay $\alpha$ (yen, dollar, etc.)\ for writing a black dot and $2\alpha$ for erasing a black dot. Then, we may define the cost function as  $c_{n}(x^n,\hat x^n) \teq \frac{1}{n} \sum_{i=1}^{n} c(x_i,\hat x_i)$, where
    \begin{align*}
        c(x, \hat x) =
        \begin{cases}
            \alpha & \mbox{ if } (x, \hat x) = (1, 0),\\
            2\alpha & \mbox{ if } (x, \hat x) = (0, 1),\\
            0 & \mbox{ otherwise}.
        \end{cases}
    \end{align*}
\end{ex}

Before we show several definitions, we introduce the limit superior and the limit inferior in probability \cite{hanspringerinformation}.
\begin{defn}[Limit superior/inferior in probability]
    For an arbitrary sequence $\bfZ = \{Z^n\}_{n = 1}^{\infty}$ of real-valued RVs, we respectively define the limit superior and the limit inferior in probability by
    \begin{align*}
        \plimsup_{n\ra\infty} Z_n &\teq \inf\left\{\A : \lim_{n\ra \infty}\Pr\left\{Z_n > \A\right\} = 0 \right\},\\
        \pliminf_{n\ra\infty} Z_n &\teq \sup\left\{\B : \lim_{n\ra \infty}\Pr\left\{Z_n < \B\right\} = 0 \right\}.
    \end{align*}
\end{defn}

We define the worst-case cost by the limit superior in probability of the cost, i.e.,
\begin{align*}
    \plimsup_{n\ra\infty} c_{n}(X^n,f_{n}(X^n,U_{M_n})).
\end{align*}
Then, we introduce two types of achievability.
\begin{defn}
    \label{defn: epsilon-weakly achievability for average cost}
    For real numbers $R, \Gamma, \epsilon \geq 0$, we say $(R, \Gamma)$ is $\epsilon$-weakly achievable in the sense of  the average cost  if and only if there exist a sequence of integers $\{M_n\}_{n=1}^{\infty}$ and a sequence of encoders $\{f_{n}\}_{n=1}^\infty$ such that
    \begin{align}
        \limsup_{n\ra\infty}\frac{1}{n}\log M_n &\leq R, \label{def: equ: size of r num}\\
        \limsup_{n\ra\infty} \frac{1}{n}I(Y^n; f_{n}(X^n,U_{M_n})) & \leq \epsilon, \label{def: equ: weak criterion average}\\
        \limsup_{n\ra\infty}{\rm E}[c_{n}(X^n,f_{n}(X^n,U_{M_n}))] &\leq \Gamma, \notag
    \end{align} 
    where $I(X;Y)$ denotes the mutual information between RVs $X$ and $Y$, and ${\rm E}[\cdot]$ denotes the expectation.
\end{defn}
\begin{defn}
    \label{defn: epsilon-weakly achievability for worst-case cost}
    For real numbers $R, \Gamma, \epsilon \geq 0$, we say $(R, \Gamma)$ is $\epsilon$-weakly achievable in the sense of  the worst-case cost if and only if there exist a sequence of integers $\{M_n\}_{n=1}^{\infty}$ and a sequence of encoders $\{f_{n}\}_{n=1}^\infty$ such that
    \begin{align}
        \limsup_{n\ra\infty}\frac{1}{n}\log M_n &\leq R, \notag\\
        \limsup_{n\ra\infty} \frac{1}{n}I(Y^n;f_{n}(X^n,U_{M_n})) & \leq \epsilon, \label{def: equ: weak criterion worst-case}\\
        \plimsup_{n\ra\infty} c_{n}(X^n,f_{n}(X^n,U_{M_n})) &\leq \Gamma. \notag
    \end{align} 
\end{defn}
We adopt the mutual information normalized by the blocklength $n$ in these definitions (i.e., \eqref{def: equ: weak criterion average} and \eqref{def: equ: weak criterion worst-case}). This is a somewhat weak criterion of independence compared with the mutual information itself (not normalized by the blocklength). The stronger version of this criterion will be considered in the later section.

Now, we define the minimum average and the minimum worst-case costs under the weak independence criterion.
\begin{defn}
    We define the minimum average cost as
    \begin{align*}
        C_{\rm a}(\epsilon, R) &\teq \inf\{\Gamma: (R, \Gamma)\mbox{ is $\epsilon$-weakly achievable}\\
        &\quad \mbox{in the sense of the average cost}\}.
    \end{align*}
\end{defn}
\begin{defn}
    We define the minimum worst-case cost as
    \begin{align*}
        C_{\rm w}(\epsilon, R) &\teq \inf\{\Gamma: (R, \Gamma)\mbox{ is $\epsilon$-weakly achievable}\\
        &\quad \mbox{in the sense of the worst-case cost}\}.
    \end{align*}
\end{defn}

\subsection{Minimum Average and Minimum Worst-Case Costs}
In this section, we characterize the minimum average and the minimum worst-case costs. To this end, for given sequences $(\bfY, \bfX, \hat \bfX)$ of RVs, we define
\begin{align*}
    I(\bfY; \hat \bfX) &\teq \limsup_{n\ra\infty} \frac{1}{n}I(Y^n; \hat X^n),\\
    \overline{H}(\hat \bfX|\bfX) &\teq \plimsup_{n\ra\infty} \frac{1}{n}\log \frac{1}{P_{\hat X^n|X^n}(\hat X^n|X^n)},\\
    c(\bfX,\hat \bfX) &\teq \limsup_{n\ra\infty}{\rm E}[c_{n}(X^n,\hat X^n)],\\
    \overline{c}(\bfX,\hat \bfX) &\teq \plimsup_{n\ra\infty} c_{n}(X^n,\hat X^n),
\end{align*}
and denote by $\bfY - \bfX - \hat\bfX$ that the Markov chain $Y^n - X^n - \hat X^n$ holds for all $n \geq 1$.

For the minimum costs under the weak independence criterion, we have the following two theorems.
\begin{thm}
    \label{thm: weak average}
    For a pair of general sources $(\bfX,\bfY)$ and any real numbers $\epsilon, R\geq 0$, we have
    \begin{align*}
        C_{\rm a}(\epsilon, R) &= \inf_{\substack{\hat\bfX: \bfY - \bfX - \hat\bfX, \\  I(\bfY; \hat\bfX) \leq \epsilon, \overline{H}(\hat\bfX|\bfX) \leq R}} c(\bfX,\hat\bfX).
    \end{align*}
\end{thm}
\begin{thm}
    \label{thm: weak worst-case}
    For a pair of general sources $(\bfX,\bfY)$ and any real numbers $\epsilon, R\geq 0$, we have
    \begin{align*}
        C_{\rm w}(\epsilon, R) &= \inf_{\substack{\hat\bfX: \bfY - \bfX - \hat\bfX, \\  I(\bfY; \hat\bfX) \leq \epsilon, \overline{H}(\hat\bfX|\bfX) \leq R}} \overline{c}(\bfX,\hat\bfX).
    \end{align*}
\end{thm}
Since proofs of theorems are rather long, we postpone these to Section \ref{sec: proofs of theorems}. The only difference of two theorems is using a function $c(\bfX,\hat\bfX)$ or $\overline{c}(\bfX,\hat\bfX)$.

According to \cite[Theorem 8 c), d), and e)]{verduhan1994gfc}, it holds that $\overline{H}(\hat \bfX|\bfX) \leq \log |\hat \cX|$. Hence, the following two corollaries follow immediately.
\begin{cor}
    \label{cor: weak average X = Y}
    When $\bfX=\bfY$ and $R\geq \log |\hat \cX|$, we have
    \begin{align*}
        C_{\rm a}(\epsilon, R) &= \inf_{\hat \bfX: I(\bfX; \hat \bfX) \leq \epsilon}
        c(\bfX,\hat \bfX).
    \end{align*}
\end{cor}
\begin{cor}
    \label{cor: weak worst X = Y}
    When $\bfX=\bfY$ and $R\geq \log |\hat \cX|$, we have
    \begin{align*}
        C_{\rm w}(\epsilon, R) &= \inf_{\hat \bfX: I(\bfX; \hat \bfX) \leq \epsilon}
        \overline{c}(\bfX,\hat \bfX).
    \end{align*}
\end{cor}

Right-hand sides of Corollaries \ref{cor: weak average X = Y} and \ref{cor: weak worst X = Y} can be regarded as the \textit{distortion-rate} function for the variable-length coding under the average distortion criterion (see, e.g., \cite[Remark 5.7.2]{hanspringerinformation}) and the maximum distortion criterion (see, e.g., the proof of \cite[Theorem 5.6.1]{hanspringerinformation}), respectively. This fact allows us to apply many results of the distortion-rate function to our study. For example, according to the proof of \cite[Theorem 5.8.1]{hanspringerinformation}, the minimum costs for stationary memoryless sources are given by the next corollary.
\begin{cor}
    \label{cor: weak cost for stationary memoryless}
    Let $\bfX=\bfY$ and $R\geq \log |\hat \cX|$. Further, let $\bfX$ be a stationary memoryless source induced by an RV $X$ on $\cX$, and $c_{n}:\cX^n\times\hat \cX^n \ra [0,\infty)$ be an additive cost function defined by
    \begin{align*}
        c_{n}(x^n,\hat x^n) \teq \frac{1}{n} \sum_{i=1}^{n} c(x_i,\hat x_i), 
    \end{align*}
    where $c:\cX\times\hat \cX \ra [0,\infty)$. Then, we have
    \begin{align*}
        C_{\rm a}(\epsilon, R) = C_{\rm w}(\epsilon, R) = \min_{\hat X: I(X; \hat X) \leq \epsilon}
        {\rm E}[c(X,\hat X)].
    \end{align*}
\end{cor}

We also consider a mixed source $\bfX$ of two sources $\bfX_1$ and $\bfX_2$ defined by
\begin{align*}
    P_{X^n}(x^n) = \alpha P_{X_1^n}(x^n) + (1-\alpha) P_{X_2^n}(x^n),
\end{align*}
where $\alpha \in [0, 1]$. According to \cite[Remark 5.10.2]{hanspringerinformation}, we have the next corollary.
\begin{cor}
    Let $\bfX=\bfY$ and $R\geq \log |\hat \cX|$. For a subadditive cost function $\tilde c_{n}:\cX^n\times\hat \cX^n \ra [0,\infty)$ that satisfies
    \begin{align*}
        \tilde c_{n + m}((x_{1}^n, x_{2}^m), (\hat x_{1}^n, \hat x_{2}^m)) \leq \tilde c_{n}(x_{1}^n, \hat x_{1}^n) + \tilde c_{m}(x_{2}^m, \hat x_{2}^m),
    \end{align*}
    let $c_{n}(x^n, \hat x^n) = \frac{1}{n} \tilde c_{n}(x^n, \hat x^n)$ and $C_{\rm a}(\epsilon, R|\bfX)$ be the minimum average cost when $\bfX=\bfY$. Then, for a mixed source $\bfX$ of two stationary sources $\bfX_{1}$ and $\bfX_{2}$, we have
    \begin{align*}
        &C_{\rm a}(\epsilon, R|\bfX)\\
        & = \inf_{\substack{(\epsilon_1,\epsilon_2) \in [0, \infty)^{2}: \\ \alpha \epsilon_1 + (1-\alpha) \epsilon_2 \leq \epsilon}} \left(\alpha C_{\rm a}(\epsilon_1, R|\bfX_1) + (1-\alpha) C_{\rm a}(\epsilon_2, R|\bfX_2)\right).
    \end{align*}
\end{cor}

\section{Minimum Costs to Erase Information under the Strong Independence Criterion}
\label{sec: strong criterion}
In this section, we consider the minimum average cost under the strong independence criterion. In order to clarify the fundamental limit of average costs, we assume that an encoder is a stochastic encoder in this section. In other words, we consider the case where the size of the uniform random number is sufficiently large. We also assume that a source is a stationary memoryless source. Then, we give a single-letter characterization of the minimum average cost and some results obtained from this characterization.

\subsection{Problem Formulation}
In this section, we define minimum average cost under the strong independence criterion.

Let $(\bfX,\bfY)$ be the pair of stationary memoryless sources, i.e., $\{(X_i,Y_i)\}_{n=1}^\infty$ be independent copies of a pair of RVs $(X,Y)$ on $\cX\times\cY$. For the sake of brevity, we simply express the sources as $(X,Y)$. Let $f_{n}:\cX^n \ra \hat \cX^n$ be a \textit{stochastic} encoder, and $c_{n}:\cX^n\times\hat \cX^n \ra [0,\infty)$ be an additive cost function as defined in Corollary \ref{cor: weak cost for stationary memoryless}, i.e., $c_{n}(x^n,\hat x^n) \teq \frac{1}{n} \sum_{i=1}^{n} c(x_i,\hat x_i)$, where $c:\cX\times\hat \cX \ra [0,\infty)$ is an arbitrary function.

The achievablility under the strong independence criterion is defined as follows.
\begin{defn}
    \label{defn: strong achievability}
    For real numbers $\Gamma, \epsilon \geq 0$, we say $\Gamma$ is $\epsilon$-strongly achievable in the sense of the average cost if and only if there exists a sequence of stochastic encoders $\{f_{n}\}_{n=1}^\infty$ such that
    \begin{align}
        \limsup_{n\ra\infty} I(Y^n ; f_{n}(X^n)) &\leq \epsilon, \label{def: equ: strong criterion}\\
        \limsup_{n\ra\infty}{\rm E}[c_{n}(X^n,f_{n}(X^n))] &\leq \Gamma, \notag
    \end{align}
    where the expectation is with respect to the sequence $X^n$ and the output of the stochastic encoder $f_{n}$.
\end{defn}
The difference from the previous section is to use the strong independence criterion in \eqref{def: equ: strong criterion}.

The minimum average cost under the strong independence criterion is defined as follows.
\begin{defn}
    We define the minimum average cost as
    \begin{align*}
        C_{\rm a}^*(\epsilon) &\teq \inf\{\Gamma: \Gamma \mbox{ is $\epsilon$-strongly achievable}\\
        &\quad \mbox{in the sense of the average cost}\}.
    \end{align*}
\end{defn}
\begin{rem}
    We only consider the average cost in this section. This is because the minimum worst-case cost coincides with the minimum average cost after all for stationary memoryless sources. This is similar to Corollary \ref{cor: weak cost for stationary memoryless}.
\end{rem}

We also consider the non-asymptotic version of the achievablity defined as follows.
\begin{defn}
    \label{defn: strong achievability finite}
    For an integer $n \geq 1$, and real numbers $\Gamma, \epsilon \geq 0$, we say $\Gamma$ is $(n, \epsilon)$-strongly achievable in the sense of the average cost if and only if there exists a stochastic encoder $f_{n}$ such that
    \begin{align}
        I(Y^n;f_{n}(X^n)) & \leq \epsilon, \label{def: equ: strong criterion finite}\\
        {\rm E}[c_{n}(X^n,f_{n}(X^n))] &\leq \Gamma. \notag
    \end{align}    
\end{defn}
\begin{rem}
    Definition \ref{defn: strong achievability finite} adopts the strong independence criterion in \eqref{def: equ: strong criterion finite}. However, this is not important in the non-asymptotic setting because this criterion is regarded as the weak criterion if we set $\epsilon$ as $n\epsilon$.
\end{rem}

The non-asymptotic minimum average cost is defined as follows.
\begin{defn}
    We define the non-asymptotic minimum average cost for a given finite blocklength $n \geq 1$ as
    \begin{align*}
        C_{\rm a}^{*}(n,\epsilon) &\teq \inf\{\Gamma: \Gamma \mbox{ is  $(n, \epsilon)$-strongly achievable}\\
        &\quad \mbox{ in the sense of the average cost}\}.
    \end{align*}
\end{defn}
\begin{rem}
    When we employ a stochastic encoder, we can give a multi-letter characterization even for general cost functions and general sources as
    \begin{align*}
        C_{\rm a}^*(\epsilon) &= \inf_{\substack{\hat\bfX: \bfY - \bfX - \hat\bfX, \\ \limsup_{n \ra \infty} I(Y^n; \hat X^n) \leq \epsilon}} c(\bfX,\hat\bfX),\\
        C_{\rm a}^*(n, \epsilon) &= \inf_{\substack{\hat X^n: Y^n - X^n - \hat X^n, \\ I(Y^n; \hat X^n) \leq \epsilon}} \mathrm{E}[c_{n}(X^n,\hat X^n)].
    \end{align*}
    However, since this characterization is quite obvious from these definitions, we focus on the single-letter characterization of basic stationary memoryless sources and additive cost functions in this paper.
\end{rem}

\subsection{Minimum Average Costs}
In this section, we give a single-letter characterization of minimum average costs $C_{\rm a}^{*}(\epsilon)$ and $C_{\rm a}^{*}(n,\epsilon)$. Since this characterization is given by employing usual information-theoretical techniques, this might not be of the main interest. However, results obtained from it are interesting and insightful. 

First of all, we show a single-letter characterization of the non-asymptotic minimum average cost $C_{\rm a}^{*}(n,\epsilon)$.
\begin{thm}
    \label{thm: strong finite length}
    For a pair of stationary memoryless sources $(X, Y)$, any integer $n\geq 1$, and any real number $\epsilon\geq 0$, we have
    \begin{align*}
        C_{\rm a}^{*}(n,\epsilon) = \min_{\substack{\hat X: Y - X - \hat X, \\ I(Y;\hat X)\leq \frac{\epsilon}{n}}} {\rm E}[c(X,\hat X)].
    \end{align*}
\end{thm}
\begin{proof}
    First, we show the converse part. If $\Gamma$ is $(n, \epsilon)$-strongly achievable in the sense of the average cost, there exists $f_{n}$ such that
    \begin{align}
        I(Y^n; \hat X^n) &\leq \epsilon,\notag\\
        \mathrm{E}[c_{n}(X^n, \hat X^n)] &\leq \Gamma,
        \label{equ: proof: thm: strong finite length, E leq G}
    \end{align}
    where $\hat X^n = f_{n}(X^n)$. We note that
    \begin{align}
        I(Y^n;\hat X^n)
        &= \sum_{i=1}^{n}I(Y_i; \hat X^n| Y^{i-1}) \notag\\
        &= \sum_{i=1}^{n}I(Y_i; \hat X^n , Y^{i-1}) \notag\\
        &\geq \sum_{i=1}^{n}I(Y_i; \hat X_i),
        \label{equ: chain rule of mutual information}
    \end{align}
    where the second equality comes from the fact that $Y_{i}$ is independent of $Y^{i - 1}$, i.e., $I(Y_{i}; Y^{i - 1}) = 0$. On the other hand, let $Q$ be an RV on $\{1,2,\cdots,n\}$ and $(Q,Y,X,\hat X)$ be RVs on $\{1,\cdots,n\}\times\cY\times\cX\times\hat \cX$ such that $P_{QYX\hat X}(i,y,x,\hat x) = \frac{1}{n} P_{Y_i X_i \hat X_i}(y,x,\hat x)$. Then, we have
    \begin{align}
        \epsilon \geq \sum_{i=1}^{n}I(Y_i; \hat X_i) = n I(Y; \hat X|Q) \geq n I(Y;\hat X), \label{equ: strong ineq of mutual info}
    \end{align}
    where the first inequality comes from \eqref{equ: chain rule of mutual information} and the last inequality comes from the fact that $Q$ is independent of $Y$. Thus, from \eqref{equ: proof: thm: strong finite length, E leq G}, we have
    \begin{align}
        \Gamma \geq \frac{1}{n} \sum_{i=1}^{n} \mathrm{E}[c(X_i, \hat X_i)] \geq \min_{\substack{\hat X: Y-X-\hat X,\\ I(Y;\hat X)\leq \frac{\epsilon}{n}}} {\rm E}[c(X,\hat X)],
        \label{equ: strong ineq of add cost}
    \end{align}
    where the last inequality comes from \eqref{equ: strong ineq of mutual info} and the fact that $Y - X - \hat X$. Since this inequality holds for any $(n, \epsilon)$-strongly achievable $\Gamma$, we have
    \begin{align*}
        C_{\rm a}^{*}(n,\epsilon) &\geq \min_{\substack{\hat X: Y-X-\hat X,\\ I(Y;\hat X)\leq \frac{\epsilon}{n}}} {\rm E}[c(X,\hat X)].
    \end{align*}   

    Next, we show the direct part. Let $\hat X$ be an RV on $\hat \cX$ such that $Y - X - \hat X$ and 
    \begin{align*}
        I(Y;\hat X)&\leq \frac{\epsilon}{n}.
    \end{align*}
    Then, the direct part is obvious, if we define the encoder as
    \begin{align*}
        f_n(x^n) = \hat x^n \mbox{ with probability } \prod_{i=1}^{n}P_{\hat X|X}(\hat x_i|x_i).
    \end{align*}
    For this encoder, we have
    \begin{align*}
        I(Y^n; f_{n}(X^n)) &= n I(Y;\hat X) \leq \epsilon,\\
        {\rm E}[c_{n}(X^n,f_{n}(X^n))] 
        & = {\rm E}[c(X,\hat X)].
    \end{align*}
    Thus, ${\rm E}[c(X,\hat X)]$ is $(n, \epsilon)$-strongly achievable for any $\hat X$ such that $Y - X - \hat X$ and $I(Y;\hat X) \leq \frac{\epsilon}{n}$. This implies that 
    \begin{align*}
        C_{\rm a}^{*}(n,\epsilon) &\leq \min_{\substack{\hat X: Y-X-\hat X,\\ I(Y;\hat X)\leq \frac{\epsilon}{n}}} {\rm E}[c(X,\hat X)].
        \qedhere
    \end{align*}
\end{proof}
\begin{rem}
    In the converse part, the single-letter characterization in the most right-hand sides of \eqref{equ: strong ineq of mutual info} and \eqref{equ: strong ineq of add cost} are largely dependent on the assumption that sources are stationary memoryless and the cost function is additive.
\end{rem}
\begin{rem}
    Since we do not use the finiteness of $\cX$, $\cY$, and $\hat \cX$,
    Theorem \ref{thm: strong finite length}  holds even if these sets are countably infinite.
\end{rem}

Next, we give a single-letter characterization of the minimum average cost $C_{\rm a}^{*}(\epsilon)$ which shows that it is impossible to reduce the minimum cost by allowing information leakage.
\begin{thm}
    \label{thm: strong}
    For a pair of stationary memoryless sources $(X, Y)$ and any $\epsilon \geq 0$, we have
    \begin{align*}
        C_{\rm a}^*(\epsilon) = \min_{\substack{\hat X: Y - X - \hat X, \\ I(Y;\hat X) = 0}}{\rm E}[c(X,\hat X)].
    \end{align*}
\end{thm}
\begin{proof}
    If $\Gamma$ is $\epsilon$-strongly achievable in the sense of the average cost, there exists $f_{n}$ such that for any $\delta > 0$ and all sufficiently large $n > 0$,
    \begin{align*}
        I(Y^n; \hat X^n) &\leq \epsilon + \delta,\\
        \mathrm{E}[c_{n}(X^n, \hat X^n)] &\leq \Gamma + \delta,
    \end{align*}
    where $\hat X^n = f_{n}(X^n)$. By noting that $\delta > 0$ is arbitrary and $\min_{\hat X: Y-X-\hat X, I(Y;\hat X)\leq \epsilon} {\rm E}[c(X,\hat X)]$ is continuous at $\epsilon = 0$ (see \ref{sec: appendix: continuity}), the rest of the proof can be done in the same way as the proof of Theorem \ref{thm: strong finite length}. Hence, we omit the details.
\end{proof}
\begin{rem}
    The finiteness of sets $\cY$ and $\hat \cX$ is necessary to show the continuity at $\epsilon = 0$ in \ref{sec: appendix: continuity}. 
\end{rem}

According to Theorem \ref{thm: strong finite length} and Theorem \ref{thm: strong}, it holds that for any $n \geq 1$ and $\epsilon \geq 0$,
\begin{align*}
    C_{\rm a}^{*}(\epsilon) = C_{\rm a}^{*}(n, 0).
\end{align*}
Hence, we only consider $C_{\rm a}^{*}(n, \epsilon)$ because $C_{\rm a}^{*}(\epsilon)$ is a special case of it.

As in the previous section, the next corollary follows immediately.
\begin{cor}
    \label{cor: X = Y for strong finite}
    When $X=Y$, we have
    \begin{align}
        C_{\rm a}^{*}(n,\epsilon) = \min_{\hat X: I(X;\hat X)\leq \frac{\epsilon}{n}}{\rm E}[c(X,\hat X)].
        \label{equ: cor: X = Y for strong finite}
    \end{align}    
\end{cor}
According to this corollary and Corollary \ref{cor: weak cost for stationary memoryless}, when $\bfX = \bfY$ and $\bfX$ is a stationary memoryless source, it holds that for any $\epsilon \geq 0$,
\begin{align*}
    C_{\rm a}(\epsilon, R) = C_{\rm w}(\epsilon, R) = C_{\rm a}^{*}(1, \epsilon).
\end{align*}

Since the right-hand side of \eqref{equ: cor: X = Y for strong finite} is the distortion-rate function, we have some closed-form expressions of the minimum cost (see.\ e.g., \cite{hanspringerinformation} and \cite{cover2006eit}).
For example, let $\cX=\hat \cX=\{0,1\}$, $P_{X}(0) = p$, and $c(x,\hat x)=\mathbf{1}\{x\neq \hat x\}$, where $p \in [0, 1/2]$ and $\mathbf{1}\{\cdot\}$ denotes the indicator function. Then, we have
\begin{align}
    C_{\rm a}^*(n, \epsilon) = h^{-1}(|h(p) - \epsilon / n|^{+}),
    \label{equ: example for finite strong}
\end{align}
where $|x|^{+} = \max\{0, x\}$, $h(p) = - p \log p - (1 - p) \log (1 - p)$, and $h^{-1}: [0, \log 2] \ra [0, 1/2]$ is the inverse function of $h$.

Furthermore, according to Corollary \ref{cor: X = Y for strong finite}, when $X = Y$, it holds that
\begin{align*}
    C_{\rm a}^*(n,0)
    = \min_{\hat x\in\hat \cX}\mathrm{E}[c(X,\hat x)]
    \teq \Gamma_{\rm min}, \quad \forall n\geq 1,
\end{align*}
where the first equality comes from the fact that $X$ and $\hat X$ are independent. Interestingly, this can be achieved by a certain \textit{deterministic} encoder as follows: Let $\tilde x = \argmintext_{\hat x\in \hat \cX}\mathrm{E}[c(X,\hat x)]$ and define an encoder $f_{n}^{(\rm r)}$ as
\begin{align*}
    f_{n}^{(\rm r)}(x^n) \teq (\tilde x,\cdots,\tilde x), \quad \forall x^n\in\cX^n.
\end{align*}
Then, this encoder achieves $C_{\rm a}^*(n,0)\ (= \Gamma_{\mathrm{min}})$, i.e., we have
\begin{align}
    I(Y^n;f_{n}^{(\rm r)}(X^n)) &= 0, \label{equ: fr achieves independence}\\
    \mathrm{E}[c_{n}(X^n,f_{n}^{(\rm r)}(X^n))]
    &= \frac{1}{n} \sum_{i = 1}^{n} \mathrm{E}[c(X_{i}, \tilde x)]\notag\\
    &= \mathrm{E}[c(X, \tilde x)] = \Gamma_{\mathrm{min}}. \label{equ: fr achieves minimum cost}
\end{align}
This means that when $X = Y$, the optimal method of erasure is to overwrite the source sequence with repeated identical symbols using $f_{n}^{(\rm r)}$. We note that $f_{n}^{(\rm r)}$ gives the minimum average cost among encoders using repeated identical symbols.

Next, we give a sufficient condition such that $C_{\rm a}^*(n,0)$ can be achieved by the encoder $f_{n}^{(\rm r)}$. Then, we show that the case where $X = Y$ is a special case of the sufficient condition. To this end, we define the \textit{weak independence} introduced by Berger and Yeung \cite{32120}.
\begin{defn}[Weak independence]
    For a pair $(X, Y)$ of RVs, let $P_{Y|X}(\cdot | x) = (P_{Y|X}(y|x): y \in \cY)$ be the $x$th row of the stochastic matrix $P_{Y|X}$. Then, we say $Y$ is \textit{weakly independent} of $X$ if the rows $P_{Y|X}(\cdot | x)$ $(x \in \cX)$ are linearly dependent.
\end{defn}
\begin{rem}
    \label{rem: weak independence}
    If $X$ is binary, then $Y$ is weakly independent of $X$ if and only if $Y$ and $X$ are independent \cite[Remark 3]{32120}.
\end{rem}

The weak independence has a useful property for independence of a triple of RVs satisfying a Markov chain. This property is shown in the next lemma.
\begin{lem}[{\cite[Theorem 4]{32120}}]
    \label{lem: weak independence and Markov chain}
    Let $\cX$, $\cY$, and $\hat \cX$ be finite sets, and $|\hat \cX| \geq 2$. Then, for a pair $(X, Y)$ of RVs, there exists an RV $\hat X$ satisfying
    \begin{enumerate}
        \item $Y - X - \hat X$
        \item $Y$ and $\hat X$ are independent
        \item $X$ and $\hat X$ are not independent
    \end{enumerate}
    if and only if $Y$ is weakly independent of $X$.
\end{lem}

Now, we give a sufficient condition.
\begin{thm}
    If $Y$ is \textit{not} weakly independent of $X$, the optimal method of erasure is to overwrite the source sequence with repeated identical symbols using $f_{n}^{(\rm r)}$, i.e., it holds that
    \begin{align*}
        I(Y^n;f_{n}^{(\rm r)}(X^n)) &= 0,\\
        {\rm E}[c_{n}(X^n,f_{n}^{(\rm r)}(X^n))] &= C_{\rm a}^*(n,0).
    \end{align*}
\end{thm}
\begin{proof}
    Since we immediately obtain that $I(Y^n;f_{n}^{(\rm r)}(X^n)) = 0$ and ${\rm E}[c_{n}(X^n,f_{n}^{(\rm r)}(X^n))] = \Gamma_{\mathrm{min}}$ (see \eqref{equ: fr achieves independence} and \eqref{equ: fr achieves minimum cost}), we only have to show that $C_{\rm a}^{*}(n, 0) = \Gamma_{\mathrm{min}}$.

    Since $Y$ is not weakly independent of $X$, there does not exist an RV $\hat X$ simultaneously satisfying three conditions in Lemma \ref{lem: weak independence and Markov chain}. This implies that for any $\hat X$ such that $Y - X - \hat X$ and $I(Y;\hat X) = 0$, it must satisfy that $I(X; \hat X) = 0$. This is because if $I(X; \hat X) > 0$, $\hat X$ simultaneously satisfies three conditions in Lemma \ref{lem: weak independence and Markov chain}.
    
    Thus, we have
    \begin{align*}        
        C_{\rm a}^{*}(n, 0) &= \min_{\substack{\hat X: Y - X - \hat X, \\ I(Y;\hat X) = 0}} {\rm E}[c(X,\hat X)]\\
        &\eqo{(a)} \min_{\substack{\hat X: Y - X - \hat X, \\ I(Y;\hat X) = 0, I(X; \hat X) = 0}} {\rm E}[c(X,\hat X)]\\
        &\eqo{(b)} \min_{\substack{\hat X: Y - X - \hat X, \\ I(Y;\hat X) = 0, I(X; \hat X) = 0}} \sum_{\hat x \in \hat \cX} P_{\hat X}(\hat x) \mathrm{E}[c(X, \hat x)]\\
        &\geq \Gamma_{\mathrm{min}},
    \end{align*}
    where (a) comes from the above argument and (b) follows since $X$ and $\hat X$ are independent.

    Since the opposite direction is obvious by setting $\hat X = \tilde x$ with probability $1$, this completes the proof.
\end{proof}
If $X = Y$, $Y$ is not weakly independent of $X$. Thus, this is a special case of this sufficient condition. According to Remark \ref{rem: weak independence}, we can also show that if $X$ is binary, the encoder $f_{n}^{(\rm r)}$ is optimal as long as $Y$ and $X$ are not independent.

On the other hand, if $Y$ is weakly independent of $X$, $C_{\rm a}^*(n,0)$ cannot be achieved by the repeated symbols using the encoder $f_{n}^{(\rm r)}$ in general. To show this fact, we give an example such that $C_{\rm a}^*(n,0) < \Gamma_{\mathrm{min}}$. Let $\cY=\{0,1\}$, $\cX = \hat \cX = \{0,1,2\}$, $c(x,\hat x) = \mathbf{1}\{x \neq \hat x\}$, $P_{X}(x) = 1/3$ for all $x \in \{0, 1, 2\}$, and
\begin{align*}
    P_{Y|X} =  
    \left[
        \begin{matrix}
            1 & 0 \\
            0 & 1\\
            0 & 1
        \end{matrix}
    \right],
\end{align*}
where the $x$th row and the $y$th column denotes the conditional probability $P_{Y|X}(y|x)$. Then, we have $\Gamma_{\rm min} = 2/3$. We note that $Y$ is weakly independent of $X$.
On the other hand, we consider an RV $\hat X$ such that $Y - X - \hat X$, and 
\begin{align*}
    P_{\hat X|X} =  
    \left[
        \begin{matrix}
            1/3 &  1/3 &  1/3 \\
            1/6 &  2/3 &  1/6 \\
            1/2 &  0 &  1/2
        \end{matrix}
    \right],
\end{align*}
where the $x$th row and the $\hat x$th column denotes the conditional probability $P_{\hat X|X}(\hat x|x)$. Then, one can easily check that $Y$ is independent of $\hat X$, and
\begin{align}
    C_{\rm a}^*(n,0) 
    \leq {\rm E}[c(X,\hat X)] 
    &= 1/2 < \Gamma_{\rm min}.
    \label{equ: fr is no longer optimal when Y is WI of X}
\end{align}
Hence, the encoder $f_{n}^{(\rm r)}$ is no longer optimal. 

Further, if we allow a little bit of leakage of information, i.e., $\epsilon > 0$, the encoder $f_{n}^{(\rm r)}$ is no longer optimal even if $Y$ is not weakly independent of $X$. This is because in general, it holds that $C_{\rm a}^*(n,\epsilon) < \Gamma_{\rm min}$ for $\epsilon > 0$ (see \eqref{equ: example for finite strong} and also \eqref{equ: fr is no longer optimal when Y is WI of X}). 

The optimality of the encoder $f_{n}^{(\rm r)}$ is summarized in Table \ref{table: optimality of fnr}.

\begin{table}[t]
    \caption{This table shows that $f_{n}^{(\rm r)}$ is optimal or not in the sense that it  whenever can achieve the minimum average cost $C_{\rm a}^*(n, \epsilon)$ or not for each corresponding condition. WI is an abbreviation for ``weakly independent''.}
    \centering
    \begin{tabular}[t]{|c|c|c|}
      \hline
      & $Y$ is not WI of $X$ & $Y$ is WI of $X$\\
      \hline
      $\epsilon = 0$ & optimal & not optimal\\
      \hline
      $\epsilon > 0$ & not optimal & not optimal\\
      \hline
    \end{tabular}
    \label{table: optimality of fnr}
\end{table}

\section{Proofs of Theorems}
\label{sec: proofs of theorems}
In this section, we prove Theorems \ref{thm: weak average} and \ref{thm: weak worst-case}.

\subsection{Fundamental Lemmas for the Random Number Generation}
In this section, we introduce some lemmas to prove Theorems \ref{thm: weak average} and \ref{thm: weak worst-case}. Since proofs of these lemmas are similar to the proofs in \cite[Section 2]{hanspringerinformation}, we will omit the details.

For two probability distributions $P$ and $Q$ on the same set $\cX$, 
we define the variational distance $d(P, Q)$ as
\begin{align*}
    d(P, Q) \teq \sum_{x \in \cX} |P(x) - Q(x)|.
\end{align*}

For all lemmas in this section, let $(\bfX, \bfY, \bfZ)=\{(X^n,Y^n,Z^n)\}_{n=1}^{\infty}$ be a triple of sequences of RVs, where $(X^n,Y^n,Z^n)$ is a triple of RVs on $\cX^n\times\cY^n\times\cZ^n$. For this triple, we define
\begin{align*}
    \cS_n(\alpha) &\teq \left\{(x^n,z^n)\in\cX^n\times\cZ^n:\right.\\
    &\left. \quad \frac{1}{n}\log\frac{1}{P_{X^n|Z^n}(x^n|z^n)}\geq \alpha\right\},\\
    \cT_n(\beta) &\teq \left\{(y^n,z^n)\in\cY^n\times\cZ^n: \right.\\
    &\left.\quad \frac{1}{n}\log\frac{1}{P_{Y^n|Z^n}(y^n|z^n)}\leq \beta \right\}.
\end{align*}

The next lemma is an extended version of \cite[Lemma 2.1.1]{hanspringerinformation}.
\begin{lem}
    \label{lem: variational distance of two RVs}
    For any integer $n \geq 1$ and any real numbers $\gamma > 0$ and $a\in\mathbb{R}$, there exists a mapping $\varphi_{n}:\cX^n\times\cZ^n\ra\cY^n$ satisfying
    \begin{align*}
        d(P_{Y^nZ^n}, P_{\tilde Y^nZ^n})
        &\leq 2\Pr\{(X^n,Z^n)\notin \cS_n(a+\gamma)\}\\
        &\quad +2\Pr\{(Y^n,Z^n)\notin \cT_n(a)\} + 2 e^{-n\gamma},
    \end{align*}
    where $\tilde Y^n = \varphi_{n}(X^n,Z^n)$.
\end{lem}
\begin{proof}
    Since this lemma can be easily proved in the same manner as the proof of \cite[Lemma 2.1.1]{hanspringerinformation}, we omit the details.
\end{proof}

The next lemma gives a sufficient condition to simulate the correlation of a pair of RVs from another RV.
\begin{lem}
    \label{lem: sufficient condition of correlation}
    If $\underline{H}(\bfX|\bfZ)>\overline{H}(\bfY|\bfZ)$,
    there exists a mapping $\varphi_{n}:\cX^n\times\cZ^n\ra\cY^n$ satisfying 
    \begin{align*} 
        \lim_{n\ra\infty} d(P_{Y^nZ^n}, P_{\tilde Y^nZ^n}) = 0,
    \end{align*}
    where $\tilde Y^n = \varphi_{n}(X^n,Z^n)$ and 
    \begin{align*}
        \underline{H}(\bfX|\bfZ) = \pliminf_{n\ra\infty} \frac{1}{n}\log \frac{1}{P_{X^n|Z^n}(X^n|Z^n)}.
    \end{align*}
\end{lem}
\begin{proof}
    Since this lemma can be easily proved by using Lemma \ref{lem: variational distance of two RVs} and the same manner as the proof of \cite[Theorem 2.1.1]{hanspringerinformation}, we omit the details.
\end{proof}

The next lemma is an extended version of \cite[Lemma 2.1.2]{hanspringerinformation}.
\begin{lem}
    \label{lem: a necessary condition}
    For any integer $n \geq 1$, any real numbers $\gamma > 0$ and $a\in\mathbb{R}$, and any mapping $\varphi_n:\cX^n\times\cZ^n\ra\cY^n$, it holds that
    \begin{align*}
        d(P_{Y^nZ^n}, P_{\tilde Y^nZ^n})
        &\geq 2 \Pr\{(Y^n, Z^n) \notin \cT_{n}(a + \gamma)\}\\
        &\quad - 2 \Pr\{(X^n, Z^n) \in \cS_{n}(a)\} - 2 e^{-n \gamma},
    \end{align*}
    where $\tilde Y^n = \varphi_{n}(X^n,Z^n)$.
\end{lem}
\begin{proof}
    Since this lemma can be easily proved in the same manner as the proof of \cite[Lemma 2.1.2]{hanspringerinformation}, we omit the details.
\end{proof}

According to this lemma, we have the next lemma which is an information spectrum version of the fact that 
\begin{align*}
    H(X | Z) \geq  H(\varphi(X, Z) | Z)
\end{align*}
for any function $\varphi$, where $H(X | Z)$ is the conditional entropy of $X$ given $Z$.
\begin{lem}
    \label{lem: decreasing of the spectrum of function}
    Let $\varphi_{n}: \cX^n \times \cZ^n \ra \cY^n$ be an arbitrary mapping and set $\tilde Y^n = \varphi_{n}(X^n, Z^n)$ and $\tilde \bfY = \{\tilde Y^n\}_{n = 1}^{\infty}$. Then, it holds that
    \begin{align*}
        \overline{H}(\bfX|\bfZ) \geq \overline{H}(\tilde \bfY|\bfZ).
    \end{align*}
\end{lem}
\begin{proof}
    Since this lemma can be easily proved by using Lemma \ref{lem: a necessary condition} and the same manner as the proof of \cite[Corollary 2.1.2]{hanspringerinformation}, we omit the details.  
\end{proof}

\subsection{Direct Part}
\label{sec: direct part}
In this section, we first show that
\begin{align}
    C_{\rm a}(\epsilon, R) &\leq \inf_{\substack{\hat\bfX: \bfY - \bfX - \hat\bfX, \\  I(\bfY; \hat\bfX) \leq \epsilon, \overline{H}(\hat\bfX|\bfX) \leq R}} c(\bfX,\hat\bfX). \label{equ: direct part}
\end{align}
In other words, we show the direct part of the proof of Theorem \ref{thm: weak average}.

For given $R$ and $\epsilon$, let $\hat\bfX$ be a sequence of RVs such that
\begin{align}
    & \bfY - \bfX - \hat\bfX, \label{equ: direc: markov condition}\\
    & \overline{H}(\hat \bfX|\bfX) \leq R, \label{equ: direc: constraint of R}\\
    & I(\bf Y;\hat \bfX) \leq \epsilon.  \label{equ: direc: constraint of I}
\end{align}
For an arbitrarily fixed $\delta > 0$, let $\{M_n\}_{n=1}^{\infty}$ be a sequence of integers such that
\begin{align}
    M_n = \left \lceil e^{n (R+\delta)} \right \rceil. \label{equ: direct: def of Mn}
\end{align}
Then, we have
\begin{align*}
    \underline{H}(\bfU|\bfX)
    = R+\delta
    > \overline{H}(\hat \bfX|\bfX),
\end{align*}
where $\bfU=\{U_{M_{n}}\}_{n=1}^\infty$ and the inequality comes from \eqref{equ: direc: constraint of R}.
Thus, according to Lemma \ref{lem: sufficient condition of correlation}, there exists a sequence of functions $f_{n}:\cX^n\times \cU_{M_n}\ra \hat \cX^n$ such that
\begin{align*}
    \lim_{n\ra\infty}d(P_{\hat X^n X^n}, P_{\tilde X^n X^n}) = 0,
\end{align*}
where $\tilde X^n = f_{n}(X^n,U_{M_n})$. Since $Y^n - X^n - \hat X^n$ and $Y^n - X^n - \tilde X^n$, we also have
\mathindent=0mm
\begin{align*}
    \lim_{n\ra\infty}d(P_{\hat X^nX^nY^n}, P_{\tilde X^nX^nY^n})
    = \lim_{n\ra\infty}d(P_{\hat X^n X^n}, P_{\tilde X^n X^n}) = 0.
\end{align*}
\mathindent=7mm
Hence from the continuity of the mutual information (see, e.g., \cite[Lemma 2.7]{csiszar2011itc}), we have
\begin{align}
    I(\bfY;\tilde \bfX) = I(\bfY;\hat \bfX) \leq \epsilon, \label{equ: direct: I tilde = I hat}
\end{align}
where $\tilde \bfX = \{\tilde X^n\}_{n = 1}^{\infty}$ and the last inequality comes from \eqref{equ: direc: constraint of I}. We also have
\begin{align}
    &c(\bfX,\tilde \bfX) - c(\bfX,\hat \bfX)\notag\\
    &=\limsup_{n\ra\infty}{\rm E}[c_{n}(X^n,\tilde X^n)]
    - \limsup_{n\ra\infty}{\rm E}[c_{n}(X^n,\hat X^n)]\notag\\
    &\leq \limsup_{n\ra\infty} ({\rm E}[c_{n}(X^n,\tilde X^n)]
    - {\rm E}[c(X^n,\hat X^n)])\notag\\
    & \leq \limsup_{n\ra\infty}d(P_{\hat X^n X^n}, P_{\tilde X^n X^n}) c_{\rm max}
    = 0.
    \label{equ: direct: limsup of cost}
\end{align}
According to \eqref{equ: direct: def of Mn}, \eqref{equ: direct: I tilde = I hat}, and \eqref{equ: direct: limsup of cost},
there exist $\{M_n\}_{n=1}^\infty$ and $\{f_{n}\}_{n=1}^\infty$ such that
\begin{align*}
    \limsup_{n\ra\infty}\frac{1}{n}\log M_n &\leq R+\delta,\\
    I(\bfY;\tilde \bfX) &\leq \epsilon,\\
    c(\bfX,\tilde \bfX) &\leq c(\bfX, \hat \bfX)
\end{align*}
for any sequence $\hat \bfX$ of RVs satisfying \eqref{equ: direc: markov condition}, \eqref{equ: direc: constraint of R}, and \eqref{equ: direc: constraint of I}.
This means that $(R+\delta, c(\bfX,\hat \bfX))$ is $\epsilon$-weakly achievable for any $\delta>0$. Then, by using the usual diagonal line argument \cite{hanspringerinformation}, we can show that $(R, c(\bfX,\hat \bfX))$ is also $\epsilon$-weakly achievable. This implies \eqref{equ: direct part}.

For the same RV $\tilde X^n = f_{n}(X^n,U_{M_n})$ as above, we have
\mathindent=0mm
\begin{align*}
    &\limsup_{n\ra\infty}\Pr\{c_{n}(X^n, \tilde X^n)>\alpha\} - \limsup_{n\ra\infty}\Pr\{c_{n}(X^n, \hat X^n)>\alpha\}\notag\\
    &\leq \limsup_{n\ra\infty} \left(\Pr\{c_{n}(X^n, \tilde X^n)>\alpha\} - \Pr\{c_{n}(X^n, \hat X^n)>\alpha\}\right)\notag\\
    &=\limsup_{n\ra\infty}\sum_{(x^n, \hat x^n) \in \cX^n \times \hat \cX^n} \left(P_{X^n\tilde X^n}(x^n, \hat x^n) - P_{X^n\hat X^n}(x^n, \hat x^n) \right)\notag\\
    &\quad \times \mathbf{1}\{c_{n}(x^n, \hat x^n)>\alpha\}\notag\\
    &\leq \limsup_{n\ra\infty}\sum_{(x^n, \hat x^n) \in \cX^n \times \hat \cX^n} |P_{X^n\tilde X^n}(x^n, \hat x^n) - P_{X^n\hat X^n}(x^n, \hat x^n)|\notag\\
    &\quad \times \mathbf{1}\{c_{n}(x^n, \hat x^n)>\alpha\}\notag\\
    &\leq \limsup_{n\ra\infty} d(P_{\hat X^n X^n}, P_{\tilde X^n X^n})
    = 0.
\end{align*}
\mathindent=7mm
Thus, we have
\mathindent=0mm
\begin{align}
    \limsup_{n\ra\infty}\Pr\{c_{n}(X^n, \tilde X^n)>\alpha\} \leq \limsup_{n\ra\infty}\Pr\{c_{n}(X^n, \hat X^n)>\alpha\}.
    \label{equ: limsup tilde cn >= limsup hat cn}
\end{align}
Similarly, we also have
\mathindent=0mm
\begin{align}
    \limsup_{n\ra\infty}\Pr\{c_{n}(X^n, \hat X^n)>\alpha\} \leq \limsup_{n\ra\infty}\Pr\{c_{n}(X^n, \tilde X^n)>\alpha\}.
    \label{equ: limsup tilde cn <= limsup hat cn}
\end{align}
\mathindent=7mm
By combining \eqref{equ: limsup tilde cn >= limsup hat cn} and \eqref{equ: limsup tilde cn <= limsup hat cn}, we have
\mathindent=0mm
\begin{align}
    \limsup_{n\ra\infty}\Pr\{c_{n}(X^n, \tilde X^n)>\alpha\} = \limsup_{n\ra\infty}\Pr\{c_{n}(X^n, \hat X^n)>\alpha\}.
    \label{equ: limsup tilde cn = limsup hat cn}
\end{align}
\mathindent=7mm
Hence, we have
\begin{align}
    \overline{c}(\bfX, \tilde \bfX)
    &= \inf \{\alpha: \lim_{n\ra\infty}\Pr\{c_{n}(X^n, \tilde X^n)>\alpha\} = 0\}\notag\\
    &= \inf \{\alpha: \lim_{n\ra\infty}\Pr\{c_{n}(X^n, \hat X^n)>\alpha\} = 0\}\notag\\
    &= \overline{c}(\bfX, \hat \bfX), \label{equ: direct: plimsup of cost}
\end{align}
where the second equality comes from \eqref{equ: limsup tilde cn = limsup hat cn}. By replacing \eqref{equ: direct: limsup of cost}, $c(\bfX, \tilde \bfX)$, and $c(\bfX, \hat \bfX)$ with \eqref{equ: direct: plimsup of cost}, $\overline{c}(\bfX, \tilde \bfX)$, and $\overline{c}(\bfX, \hat \bfX)$, respectively, and repeating the same argument as above, we also have
\begin{align*}
    C_{\rm w}(\epsilon, R) &\leq \inf_{\substack{\hat\bfX: \bfY - \bfX - \hat\bfX, \\  I(\bfY; \hat\bfX) \leq \epsilon, \overline{H}(\hat\bfX|\bfX) \leq R}} \overline{c}(\bfX,\hat\bfX).
\end{align*}
This is the direct part of the proof of Theorem \ref{thm: weak worst-case}.

\begin{rem}
    Since $I(\bfY;\tilde \bfX)$ and $I(\bfY;\hat \bfX)$ are mutual information normalized by the blocklength, the first equality in \eqref{equ: direct: I tilde = I hat} holds by using the continuity. However, for the mutual information itself, the equality $\limsup_{n \ra \infty} I(Y^n; \tilde X^n) = \limsup_{n \ra \infty} I(Y^n; \hat X^n)$ is no longer guaranteed. Thus, the above proof may be invalid under the strong independence criterion. This is one of the reasons why we employ a stochastic encoder in Section \ref{sec: strong criterion}.
\end{rem}
\begin{rem}
    Since \cite[Lemma 2.7]{csiszar2011itc} holds only for finite sets, the finiteness of sets $\cY$ and $\hat \cX$ is necessary to show the first equality in \eqref{equ: direct: I tilde = I hat}. If $\cY$ and $\hat \cX$ are countably infinite sets and the equality holds even for these sets, the direct part also holds for these sets. We also note that the finiteness of $\cX$ is actually unnecessary.
\end{rem}

\subsection{Converse Part}
\label{sec: converse part}
In this section, we first show that
\begin{align}
    C_{\rm a}(\epsilon, R) &\geq \inf_{\substack{\hat\bfX: \bfY - \bfX - \hat\bfX, \\ I(\bfY; \hat\bfX) \leq \epsilon, \overline{H}(\hat\bfX|\bfX) \leq R}} c(\bfX,\hat\bfX). \label{equ: conv part}
\end{align}
In other words, we show the converse part of the proof of Theorem \ref{thm: weak average}.

If $(R, \Gamma)$ is $\epsilon$-weakly achievable, there exist sequences of integers $\{M_n\}_{n=1}^{\infty}$ and encoders $\{f_{n}\}_{n=1}^\infty$ such that
\begin{align}
    \limsup_{n\ra\infty}\frac{1}{n}\log M_n &\leq R, \label{equ: conv: inequality of rate}\\
    I(\bfY;\hat \bfX) &\leq \epsilon, \label{equ: conv: inequality of epsilon}\\
    c(\bfX,\hat \bfX) &\leq \Gamma, \label{equ: conv: inequality of cost}
\end{align}
where $\hat \bfX = \{\hat X^n\}_{n=1}^\infty$ and $\hat X^n=f_{n}(X^n,U_{M_n})$.

According to Lemma \ref{lem: decreasing of the spectrum of function}, we have
\begin{align*}
    \overline{H}(\bfU|\bfX) \geq \overline{H}(\hat \bfX|\bfX).
\end{align*}
On the other hand, due to \eqref{equ: conv: inequality of rate}, we have
\begin{align*}
    \overline{H}(\bfU|\bfX) \leq R.
\end{align*}
Thus, we have
\begin{align}
    \overline{H}(\hat \bfX|\bfX) \leq R.
    \label{equ: conv: condition of rate}
\end{align} 

Now, by combining \eqref{equ: conv: inequality of epsilon}, \eqref{equ: conv: inequality of cost}, \eqref{equ: conv: condition of rate}, and the fact that $\hat \bfX$ satisfies $\bfY - \bfX - \hat \bfX$, we have
\begin{align*}
    \Gamma \geq c(\bfX,\hat \bfX)
    \geq \inf_{\substack{\hat\bfX: \bfY - \bfX - \hat\bfX, \\  I(\bfY; \hat\bfX) \leq \epsilon, \overline{H}(\hat\bfX|\bfX) \leq R}} c(\bfX,\hat\bfX).
\end{align*}
Since this inequality holds for any $\Gamma$ such that $(R,\Gamma)$ is $\epsilon$-weakly achievable, we have \eqref{equ: conv part}.

By replacing $c(\bfX,\hat\bfX)$ with $\overline{c}(\bfX,\hat\bfX)$ and repeating the same argument as above, we also have
\begin{align}
    C_{\rm w}(\epsilon, R) &\geq \inf_{\substack{\hat\bfX: \bfY - \bfX - \hat\bfX, \\ I(\bfY; \hat\bfX) \leq \epsilon, \overline{H}(\hat\bfX|\bfX) \leq R}} \overline{c}(\bfX,\hat\bfX). \label{equ: conv part worst}
\end{align}
This is the converse part of the proof of Theorem \ref{thm: weak worst-case}.

\begin{rem}
    Unlike the direct part, we do not use the continuity of the mutual information in the converse part. Thus, the proof of this part is valid even if we adopt the strong independence criterion.
\end{rem}
\begin{rem}
    Since we do not use the finiteness of sets $\cX$, $\cY$, and $\hat \cX$ in the converse part, this part holds even if these sets are countably infinite.
\end{rem}

\section{Conclusion}
\label{sec: conclusion}
In this paper, we introduced the information erasure model and considered minimum costs under the weak and the strong independence criteria. For the weak independence criterion, we characterized the minimum average and the minimum worst-case costs for general sources by using information-spectrum quantities. On the other hand, for the strong independence criterion, we gave a single-letter characterization of the minimum average cost for stationary memoryless sources. By using this characterization, we gave a sufficient condition such that the optimal method of erasure is to overwrite the source sequence with repeated identical symbols.

\section*{Acknowledgment}
The authors would like to thank Prof.~H.~Yamamoto for teaching us his result \cite{1056749}, Prof.~Y.~Oohama for teaching us the paper \cite{32120}, Prof.~H.~Yagi for teaching us the paper \cite{8007053}, and the anonymous reviewers for their valuable comments.

\appendix
\def\thesection{Appendix \Alph{section}}
\section{Continuity at $\epsilon = 0$}
\label{sec: appendix: continuity}
In this appendix, we show that
\mathindent=0mm
\begin{align}
    \lim_{\epsilon \downarrow 0} \min_{\substack{\hat X: Y-X-\hat X,\\ I(Y;\hat X)\leq \epsilon}} \mathrm{E}[c(X,\hat X)] = \min_{\substack{\hat X: Y-X-\hat X,\\ I(Y;\hat X) = 0}} \mathrm{E}[c(X,\hat X)].
    \label{equ: appendix: disired equality}
\end{align}
\mathindent=7mm

Let $\{\epsilon_n\}_{n = 1}^{\infty}$ be a sequence such that $\epsilon_{n} > 0$ and $\epsilon_n \ra 0$. Then, we have
\begin{align}
    \lim_{\epsilon \downarrow 0} \min_{\substack{\hat X: Y-X-\hat X,\\ I(Y;\hat X)\leq \epsilon}} \mathrm{E}[c(X,\hat X)] = \lim_{n \ra \infty} \min_{\substack{\hat X: Y-X-\hat X,\\ I(Y;\hat X)\leq \epsilon_{n}}} \mathrm{E}[c(X,\hat X)].
    \label{equ: appendix: e down 0 = lim en}
\end{align}

Let $P_{\hat X^{(n)}|X}: \cX \ra \hat \cX$ be a conditional probability distribution such that
\begin{align}
    \mathrm{E}[c(X,\hat X^{(n)})] &= \min_{\substack{\hat X: Y-X-\hat X,\\ I(Y;\hat X)\leq \epsilon_n}} {\rm E}[c(X,\hat X)],
    \label{equ: appendix: Pn achieves min cost}\\
    I(Y;\hat X^{(n)}) &\leq \epsilon_n.
    \label{equ: appendix: mutual condition}
\end{align}
Then, for the sequence $\{P_{\hat X^{(n)}|X}\}_{n = 1}^{\infty}$, there exists a convergent subsequence $\{P_{\hat X^{(n_{k})}|X}\}_{k = 1}^{\infty}$ such that $P_{\hat X^{(n_k)}|X} \ra P_{\tilde X|X}$ ($k \ra \infty$), where $P_{\tilde X|X}: \cX \ra \hat \cX$ is also a conditional probability distribution. Then, by the continuity, we have
\begin{align*}
    \mathrm{E}[c(X,\tilde X)]
    &= \lim_{k \ra \infty} \mathrm{E}[c(X, \hat X^{(n_k)})]\\
    &\eqo{(a)} \lim_{\epsilon \downarrow 0} \min_{\substack{\hat X: Y-X-\hat X,\\ I(Y;\hat X)\leq \epsilon}} \mathrm{E}[c(X,\hat X)],
\end{align*}
and
\begin{align*}
    I(Y; \tilde X) &\eqo{(b)} \lim_{k \ra \infty} I(Y; \hat X^{(n_k)})\\
    &\leqo{(c)} \lim_{k \ra \infty} \epsilon_{n_k} = 0,
\end{align*}
where (a) comes from \eqref{equ: appendix: e down 0 = lim en} and \eqref{equ: appendix: Pn achieves min cost}, (b) comes from \cite[Lemma 2.7]{csiszar2011itc} and the finiteness of $\cY$ and $\hat \cX$, and (c) comes from \eqref{equ: appendix: mutual condition}. Thus, we have
\mathindent=0mm
\begin{align}
    \min_{\substack{\hat X: Y-X-\hat X,\\ I(Y;\hat X) = 0}} \mathrm{E}[c(X,\hat X)]
    &\leq \mathrm{E}[c(X,\tilde X)] \notag\\
    &= \lim_{\epsilon \downarrow 0} \min_{\substack{\hat X: Y-X-\hat X,\\ I(Y;\hat X)\leq \epsilon}} \mathrm{E}[c(X,\hat X)].
    \label{equ: appendix: e down 0 leq lim en}
\end{align}
\mathindent=7mm

Since the opposite direction is obvious, we have \eqref{equ: appendix: disired equality} from \eqref{equ: appendix: e down 0 leq lim en}.

\end{document}

%% file: myedefinition.tex
\usepackage{amsmath}
\usepackage{amssymb}
\usepackage{amsthm}
\theoremstyle{definition}
\newtheorem{thm}{Theorem}
\newtheorem{lem}{Lemma}
\newtheorem{cor}{Corollary}
\newtheorem{defn}{Definition}
\newtheorem{rem}{Remark}
\newtheorem{ex}{Example}
\newcommand{\ra}{\rightarrow}

\newcommand{\teq}{\triangleq}
\newcommand{\A}{\alpha}
\newcommand{\B}{\beta}

\newcommand{\cS}{\mathcal{S}}
\newcommand{\cT}{\mathcal{T}}
\newcommand{\cU}{\mathcal{U}}

\newcommand{\cX}{\mathcal{X}}
\newcommand{\cY}{\mathcal{Y}}
\newcommand{\cZ}{\mathcal{Z}}

\newcommand{\argmintext}{\mathop\mathrm{argmin}\nolimits}
\newcommand{\plimsup}{\mathrm{p}\mathchar`-\!\mathop{\limsup}\limits}
\newcommand{\pliminf}{\mathrm{p}\mathchar`-\!\mathop{\liminf}\limits}

\newcommand{\bfU}{\mathbf{U}}

\newcommand{\bfX}{\mathbf{X}}
\newcommand{\bfY}{\mathbf{Y}}
\newcommand{\bfZ}{\mathbf{Z}}
\def\eqo#1{\overset{\mathrm{#1}}{=}}

\def\leqo#1{\overset{\mathrm{#1}}{\leq}}

%% file: iemodel.tex
\begin{figure}[t]
\small
\centering
\begin{picture}(250, 70)
 \put(0, 35){source 1 seq.}
 \put(10, 20){$Y^n$}

 \put(40, 22){\vector(-1, 0){15}}
 \put(42, 20){correlated}
 \put(86, 22){\vector(1, 0){15}}

 \put(85, 35){source 2 seq.}
 \put(105, 20){$X^n$}

 \put(120, 22){\vector(1, 0){20}}

 \put(140, 10){\framebox(40, 30){}}
 \put(145, 30){encoder}
 \put(155, 15){$f_{n}$}

 \put(160, 55){\vector(0, -1){15}}

 \put(54, 60){uniform random number}
 \put(150, 60){$U_{M_n}$}

 \put(180, 22){\vector(1, 0){20}}

 \put(195, 35){output}
 \put(205, 20){$\hat X^n$}

 \put(180, 75){cost function}
 \put(180, 60){$c_{n}(X^n,\hat X^n)$}

 \put(13, 2){\vector(0, 1){15}}
 \put(13, 2){\line(1, 0){25}}
 \put(40, 0){independent}
 \put(92, 2){\line(1, 0){118}}
 \put(210, 2){\vector(0, 1){15}}
\end{picture}
 \caption{Information Erasure Model}
 \label{fig:iemodel}
\end{figure}
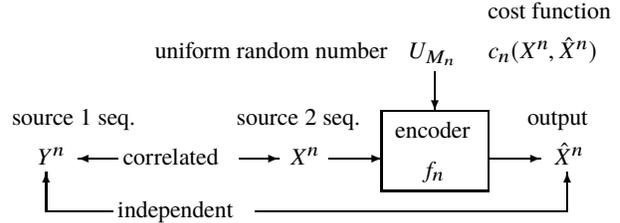
